\documentclass{ifacconf}

\usepackage{graphicx}      
\usepackage{natbib}  
\usepackage[utf8]{inputenc}
\usepackage{amsmath,amsfonts}
\usepackage{setspace}
\usepackage{algorithm}
\usepackage{graphicx}
\usepackage{mathtools}

\renewcommand{\cite}{\citep}
\renewenvironment{bmatrix}{\bsmallmatrix}{\endbsmallmatrix}

\newtheorem{theorem}{Theorem}

\newcommand{\agentSet}{\mathcal{A}}
\newcommand{\edgeSet}{\mathcal{E}}

\newcommand{\neighborSet}{\mathcal{N}}

\newcommand{\greed}{\gamma}

\newcommand{\numAgents}{n}
\newcommand{\numLanes}{m}

\newcommand{\numSlices}{\sigma}

\newcommand{\balVec}{\mathbf{b}}
\newcommand{\bal}{b}

\newcommand{\trans}{T}
\newcommand{\transAmount}{a}
\newcommand{\transVec}{\boldsymbol{\tau}}
\newcommand{\transElem}{\tau}

\newcommand{\transLaneSet}{{\laneSet_{\trans}}}

\newcommand{\augTransLaneSet}{{\bar{\laneSet}_{\trans}}}
\newcommand{\optTransLaneSet}{\laneSet^*_{\trans}}

\newcommand{\transPool}{{\mathcal{P}}}

\newcommand{\lane}{s}
\newcommand{\laneSet}{\mathcal{S}}

\newcommand{\feasLaneSets}{\mathbb{S}}
\newcommand{\laneUseVec}{\mathbf{u}}
\newcommand{\laneUseElem}{u}
\newcommand{\lanePool}{{\mathcal{P}}}

\newcommand{\congElem}{\lambda}
\newcommand{\congVec}{\boldsymbol{\congElem}}

\newcommand{\laneEff}{B}
\newcommand{\transEff}{F}

\newcommand{\transPrice}{\phi}
\newcommand{\priceFunc}{f}

\newcommand{\cardExp}{\alpha}

\newcommand{\opPrice}{p_0}
\newcommand{\maxFee}{\transPrice_{\max}}

\newcommand{\reals}{\mathbb{R}}

\newcommand{\elVec}{\mathbf{e}}

\newcommand{\simplex}{\Delta}
\newcommand{\pureStrategySet}{\Pi}

\newcommand{\laneReqWeight}{w}
\newcommand{\laneReqWeightVec}{\mathbf{\laneReqWeight}}
\newcommand{\laneSendWeight}{v}
\newcommand{\laneSendWeightVec}{\mathbf{\laneSendWeight}}
\newcommand{\laneSendWeightEstVec}{\hat{\laneSendWeightVec}}

\newcommand{\laneReqWeightEstVec}{\hat{\laneReqWeightVec}}

\newcommand{\brLaneReqWeightVec}{\laneReqWeightVec^*}
\newcommand{\brEstLaneReqWeightVec}{\laneReqWeightVec^+}

\newcommand{\brLaneSendWeightVec}{\laneSendWeightVec^*}
\newcommand{\brEstLaneSendWeightVec}{\laneSendWeightVec^+}

\newcommand{\potMatrix}{P}
\newcommand{\potMatrixCard}{\potMatrix_{\text{card}}}
\newcommand{\potMatrixBal}{\potMatrix_{\text{bal}}}
\newcommand{\potMatrixEff}{\potMatrix_{\text{eff}}}
\newcommand{\potFunc}{H}
\newcommand{\action}{a}
\newcommand{\actionSet}{\mathcal{S}}
\newcommand{\allActions}{\boldsymbol{\mathcal{S}}}
\newcommand{\utilFun}{u}

\begin{document}

\begin{frontmatter}

\title{Transaction Pricing for Maximizing Throughput in a Sharded Blockchain Ledger} 

\author[First]{James R. Riehl} 
\author[First]{Jonathan Ward} 

\address[First]{Fetch.ai, St. John's Innovation Centre, Cowley Road, Cambridge, CB4 0WS, UK (e-mail: \{james.riehl, jonathan.ward}@fetch.ai).\}

\begin{abstract}
In this paper, we present a pricing mechanism that aligns incentives of agents who exchange resources on a decentralized ledger with the goal of maximizing transaction throughput. Subdividing a blockchain ledger into shards promises to greatly increase transaction throughput with minimal loss of security. However, the organization and type of the transactions also affects the ledger's efficiency, which is increased by wallet agents transacting in a single shard whenever possible while collectively distributing their transactions uniformly across the available shards.
Since there is no central authority to enforce these properties, the only means of achieving them is to design the system such that it is in agents' interest to act in a way that benefits overall throughput.
We show that our proposed pricing policy does exactly this by inducing a potential game for the agents, where the potential function relates directly to ledger throughput.
Simulations demonstrate that this policy leads to near-optimal throughput under a variety of conditions.
\end{abstract}


\end{frontmatter}

\section{Introduction} \label{intro}

Decentralized ledgers, commonly implemented as encrypted linked lists of transaction records, or \textit{blockchains}, allow individuals to trade resources and maintain a common state machine securely and without a central authority \cite{narayanan2016bitcoin,crosby2016blockchain}.
As the demand for such systems grows, the slow throughput of established systems, e.g. 7-15 transactions per second on Bitcoin and Ethereum \cite{croman2016scaling}, is becoming a major obstacle to more widespread adoption and success, especially in applications requiring high frequency or time-critical transactions.
One innovation that promises to significantly increase transaction throughput is subdividing each the transaction records into distinct \textit{shards}, allowing parallel communication, execution and storage of transactions that use different shards \cite{saraph2019empirical}.
These subdivision methods have proven effective in conventional databases \cite{corbett2013spanner}.
By evenly distributing transactions across multiple shards, the system can process transactions much faster than a serial blockchain ledger.
However, in a decentralized system, different external users submit transactions, and there is no guarantee that they will choose to distribute their transactions in a way that enables the system to operate efficiently.
In this paper, we investigate the use of small transaction surcharges as incentives to align the goals of the individual users with the system-wide goal of maximizing throughput.

Decentralized ledgers pose additional challenges due to various stakeholders having different and often competing objectives.
For example, it is in the agents' interest that transactions are fast and cheap, which is more likely to occur when congestion is low, but \textit{miners} or \textit{validators} (agents responsible for reaching a consensus on which transactions are valid) benefit when transactions are expensive, which may be the case when congestion is high. 
These conflicting incentives must be accounted for in the design of an effective sharding system for blockchain ledgers.

Several proposals for sharding block chains have previously been put forth \cite{kokoris2018omniledger,luu2016secure,zamani2018rapidchain,buterin2016ethereum}, each of which aims to make the ledgers more scalable while maintaining appropriate levels of security.
However, these approaches all rely on randomization for distributing transactions among shards and do little to explicitly mitigate the problems of load imbalance and frequent cross-shard transactions, which consume resources in communication between shards and potentially, depending on the sharding implementation, force the pausing of execution threads on one or both shards.

We address these problems here with a transaction pricing policy that incentivizes agents to choose shards in a way that maximizes ledger throughput.
In particular, the proposed pricing function is based on a novel transaction efficiency measure that induces a potential game for the agents, where the potential function relates directly to overall transaction throughput.

The problem considered here resembles that of congestion games, a classic example of potential games, in which the goal is to minimize congestion in transportation or communication networks, for example, by setting prices to align agent incentives with this goal \cite{monderer1996potential}.
Indeed if the only goal were to minimize congestion on the transaction network, this would be a straightforward application of congestion game theory.
However, the reduced efficiency caused by cross-shard transactions introduces additional complexity that must be accounted for in the pricing mechanism.
By combining transaction size and the distribution of transactions across shards into a single quantity called \textit{transaction efficiency}, we are able to express the throughput objective as a function of agents' shard choices and previous transactions, which we then use to set the price.
Decentralized ledgers sometimes use transaction fees as incentives for other agents to maintain the ledger by validating transactions, so our proposal would simply weight such fees to promote overall efficiency.
\section{Ledger and transaction model} \label{model}

The model consists of a set of agents who transact with their neighbors in the network via a blockchain ledger.

\subsection{Agents and network}

The network consists of a set $\agentSet$ of $\numAgents$ agents who are interconnected by the edges $\edgeSet \subseteq \agentSet^2$, where an edge $(i,j) \in \edgeSet$ means that agent $i$ can request a transaction from agent $j$.
We denote the set of neighboring agents from which agent $i$ requests transactions by $\neighborSet_i := \{j \in \agentSet: (i,j) \in \edgeSet\}$. 
Each agent $i$ maintains a balance of resources $\balVec_i := [\bal_i^1,\dots,\bal_i^\numLanes]^\top$, where $\bal_i^\lane$ denotes the amount of resources agent $i$ owns in shard $\lane \in \laneSet := \{1,\dots,\numLanes\}$, and its objective is to choose shards and execute transactions in a way that minimizes transaction fees.

\subsection{Transactions}

Assume that transactions are fully asynchronous and arrive in a sequence where one agent ($i$, the \textit{receiver}) requests a transaction of an amount ($\transAmount$) from another agent ($j$, the \textit{sender}), such that the sender is in the neighbor set of the receiver ($j \in \neighborSet_i$).
Let $\trans := (i,j,\transVec)$ denote a transaction requested by agent $i$ from agent $j$, where $\transVec := [\transElem_1,\dots,\transElem_\numLanes]^\top$ lists the amounts to be transferred in each shard, which we assume are non-negative (each $\transElem_{\lane} >= 0$), and sum to the total amount ($\sum_{\lane=1}^m \transElem_{\lane} = \transAmount$).
The \textit{cardinality} $|\transLaneSet|$ is the number of shards used in the transaction. 
Let $\transPool$ denote the set (or \textit{pool}) of transactions that are waiting to be added to a block.
We define the price of a transaction as a function of the transaction itself and the current transaction pool: $\transPrice := \priceFunc(\trans,\transPool)$, to be given in precise terms in section \ref{game}.
The transaction process proceeds as follows:
\begin{enumerate}
    \setlength\itemsep{1em}
    \item \textbf{Request:} Agent $i$ requests a transaction from agent $j\in\neighborSet_i$ and specifies an amount $\transAmount$ and a shard $\lane_i \in \laneSet$ in which to receive the transaction, with the goal of minimizing current and expected future transaction fees.
    Only the sender pays the transaction fees, but the receiver has an incentive to minimize these in order to maximize the probability that the sender accepts and fulfills the transactions.
    We can express the receiver's shard choice as
    \begin{equation} \label{eq:laneOptRec}
        \lane_i \in \arg\min_{\lane \in \laneSet} \big[\greed_r \priceFunc(\trans,\transPool) + (1 - \greed_r)E[\transPrice_{ij};\lane_i]\big],
    \end{equation} 
    where $E[\transPrice_{ij};\lane_i]$ denotes the expected price of future transactions from the sender to the receiver if the receiver requests the current transaction in shard $\lane_i$, and $\greed_r \in [0,1]$ denotes the priority receiving agents place on the current transaction fee relative to future transactions.
    Estimating $E[\transPrice_{ij};\lane_i]$ is a key element in the pricing mechanism and is discussed further in section \ref{sec:efficiency}.
    \item \textbf{Fulfillment:} Agent $j$ accepts and fulfills the transaction if there are sufficient funds (including fees) and chooses a set of shards $\laneSet_j \subseteq \laneSet$ from which to send resources, with the goal of minimizing current and expected future transaction fees. 
    Since the final transaction always includes the receiver's requested shard $\lane_i$, the transaction set is given by $\transLaneSet :=  \laneSet_j \cup \{\lane_{i}\}$.
    The set of all feasible transaction shard sets is then
    \begin{equation*}
        \feasLaneSets_j(\transAmount) := \Big\{\transLaneSet \subseteq \laneSet: \sum_{\lane \in \laneSet_j} \bal_{j}^\lane \geq \transAmount + \priceFunc(\trans,\transPool)\Big\}.
    \end{equation*}
    For a requested transaction, we assume that the sender chooses a shard set that minimizes the following expression:
    \begin{equation} \label{eq:laneOptSend}
    \laneSet_{j} \in \arg\min_{\transLaneSet \in \feasLaneSets_j(\transAmount)} \big[\greed_s \priceFunc(\augTransLaneSet,\transPool) + (1 - \greed_s)E[\transPrice_{ij};\lane]\big],
    \end{equation} 
    where $\greed_s \in [0,1]$ denotes the priority sending agents place on the current transaction fee relative to future transactions.
    The sender withdraws the resources to be transferred plus transaction fees from the shards in $\optTransLaneSet$, and the receiver adds the transferred balance to shard $\lane_i$.
    The result is a  transaction $\trans$ 
    that goes into the transaction pool $\transPool \leftarrow \transPool \cup \{\trans\}$.
    \item \textbf{Block assembly:} 
    
    We assume that the blockchain is divided into $\numLanes$ shards, each of which contains $\numSlices$ slots, and that the blocks in each shard are produced synchronously. 
    
    Let $\lanePool_{\lane}$ denote the set of transactions in pool $\transPool$ that use shard $\lane$: 
    \begin{equation*}
        \lanePool_{\lane} := \{(i,j,\transVec) \in \transPool: \transElem_{\lane} > 0 \}.
    \end{equation*}
    In this paper, we assume that when any shard becomes full (i.e., there exists a shard $\lane$ such that $|\lanePool_{\lane}| > \numSlices$) the block is assembled from all transactions in the pool and $\transPool$ is reset to empty.
    The maximum theoretical capacity of the blockchain ($\numLanes \times \numSlices$) is reached when the cardinality of all  transactions is one (no cross-shard transactions) and each shard contains exactly the maximum number of transactions.
    Note that it may not be possible to execute both sides of multi-shard transactions leading to failure of transactions of this type. Therefore, assuming that all transactions in the pool are executed, as we do in this study, leads to optimistic estimates of the throughput in the case of frequent cross-shard transactions. This leads to conservative estimates of the performance gains that would arise from implementing our proposed pricing policy. 
    Although this serves as a reasonable approximation for our purposes, coordinating states between shards is a nontrivial problem and provides further motivation to incentivize single-shard transactions. 
    We note that this analysis generalizes to blockchains with complex state execution rules such as smart contracts where cross-shard transactions require at least twice the computation of an otherwise identical single-sharded transaction. 
\end{enumerate}

\section{Transaction throughput}

\begin{figure*}[ht]
    \centering
    \includegraphics[width=0.75\textwidth]{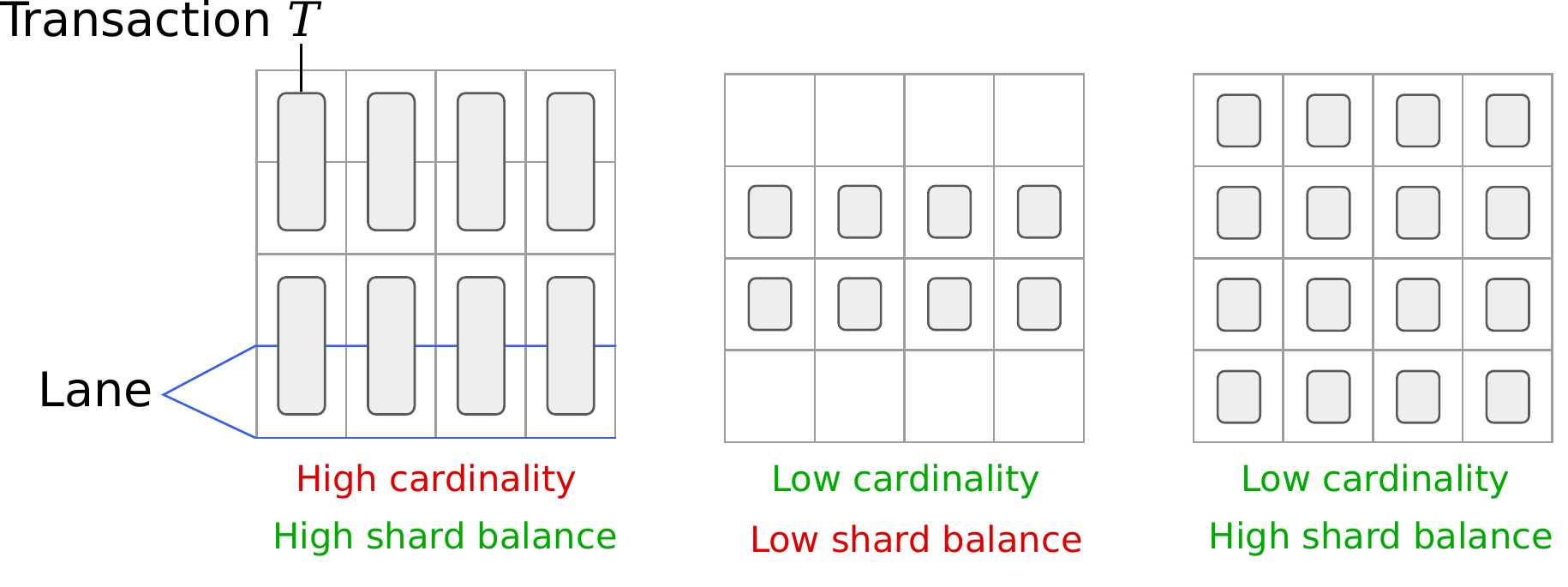}
    \caption{Cardinality and shard balance in three different transaction pools assembled into blocks. Rows corresponds to shards and columns represent block slices. The gray boxes symbolize transactions on the underlying shards.}
    \label{fig:cardcong}
\end{figure*}

Transaction throughput measures the number of transactions processed in a given unit of time, but since we do not explicitly include time in our model, we seek a time-independent alternative quantity.
Specifically, we define the \textit{transaction efficiency} of a pool $\transPool$ as the fraction of the theoretical maximum number of transactions that could be included in a block in which all transactions have cardinality one and are evenly distributed across shards.

There are two primary factors that determine transaction efficiency: \textit{cardinality} and \textit{shard  balance}.
We say that shard balance is high when the transaction pool uses the shards in roughly equal proportions and low when the transaction pool uses some shards much more than others.
Figure \ref{fig:cardcong} illustrates three partially completed blocks from transaction pools exhibiting varying degrees of cardinality and shard balance.

We measure shard balance in terms of the deviation from uniform shard usage in the transaction pool.
The usage of each shard relative to the pool is given by $\laneUseVec(\transPool) := [\laneUseElem_1(\transPool),\dots,\laneUseElem_\numLanes(\transPool)]^\top$, where
\begin{equation*}
    \laneUseElem_{\lane}(\transPool) := \frac{|\lanePool_{\lane}|}{\sum_{\trans\in\transPool}|\transLaneSet|}.
\end{equation*}
Since we consider only positive transaction fees, we focus on those shards with greater usage than average and define the loading of shard $\lane\in\laneSet$ by
\begin{equation*}
    \congElem_{\lane}(\transPool) := \max\left(0,\laneUseElem_{\lane}(\transPool) - \frac{1}{\numLanes}\right).
\end{equation*}
These values are collected in the loading vector $\congVec(\transPool) := [\congElem_1(\transPool),\dots,\congElem_\numLanes(\transPool)]^\top$, and we can now quantify the shard balance in the transaction pool as
\begin{equation} \label{eq:congestion}
    \laneEff_\transPool(\laneSet) := 1 - \sum_{\lane\in \laneSet} \congElem_{\lane}(\transPool). 
\end{equation}
Note that by definition, $\laneEff_\transPool(\laneSet) \in [0,1)$.
We can now express the total transaction efficiency as the shard balance divided by the mean cardinality of all transactions in the pool:
\begin{equation} \label{eq:transEff}
    \transEff_\transPool := \frac{\laneEff_\transPool(\laneSet)}{\displaystyle\frac{1}{|\transPool|}\sum_{\trans\in\transPool} |\transLaneSet|}.
\end{equation}
As an example, for the middle block in Figure \ref{fig:cardcong}, the shard usage is $\laneUseVec(\transPool) = (\frac{0}{8},\frac{4}{8},\frac{4}{8},\frac{0}{8})^\top$, resulting in the loading vector $\congVec(\transPool) = (0,\frac{1}{4},\frac{1}{4},0)^\top$ and total shard balance $\laneEff_\transPool(\laneSet) = \frac{1}{2}$.
Since the mean cardinality is one, the transaction efficiency is $\frac{1}{2}$. 
For the left block in Figure \ref{fig:cardcong}, since the shard balance is one and the mean cardinality is two, the transaction efficiency is also $\frac{1}{2}$.
The block on the right achieves the maximum transaction efficiency of 1.

\subsection{Expected transaction efficiency} \label{sec:efficiency}

We introduce for each edge $(i,j) \in \edgeSet$ along which a transaction can take place, a \textit{shard  request distribution} $\laneReqWeightVec_{ij} := [\laneReqWeight_{ij1},\dots,\laneReqWeight_{ij\numLanes}]^\top$, where $\laneReqWeight_{ij\lane}$ denotes the probability that agent $i$ will choose shard $\lane$ when requesting a transaction from agent $j$.
Similarly, we define a \textit{shard  sending distribution} $\laneSendWeightVec_{ij} := [\laneSendWeight_{ij1},\dots,\laneSendWeight_{ij\numLanes}]^\top$, where $\laneSendWeight_{ij\lane}$ denotes the probability that agent $j$ will choose shard $\lane$ when sending a transaction to agent $i$.
These randomized distributions, which lie on a probability simplex (each $\laneReqWeight_{ij\lane} \geq 0$ and $\sum_{\lane=1}^\numLanes \laneReqWeight_{ij\lane} = 1$, and similarly for $\laneSendWeightVec_{ij}$), model the initial uncertainty about the shards used by neighboring agents and how such uncertainty evolves toward deterministic choices.

\textbf{Expected cardinality:} We can now write the expected cardinality of transactions between agents $i$ and $j$ as:
\begin{align} \nonumber
    E\left[|\transLaneSet|\right] &= \laneReqWeightVec_{ij}^\top \potMatrixCard \laneSendWeightVec_{ij} \\ &= 
    \begin{bmatrix}
        \laneReqWeight_{ij1} & \laneReqWeight_{ij2} & \cdots & \laneReqWeight_{ij\numLanes}
    \end{bmatrix}
    \begin{bmatrix}
        1  & 2 & \cdots & 2 \\ 
        2  & 1 & \cdots & 2 \\
        \vdots & \vdots & \ddots & \vdots \\
        2 & 2  & \cdots & 1 
    \end{bmatrix}
    \begin{bmatrix}
        \laneSendWeight_{ij1} \\ \laneSendWeight_{ij2} \\ \vdots \\ \laneSendWeight_{ij\numLanes}
    \end{bmatrix},   
    \label{eq:expCard}
\end{align}
where $\potMatrixCard$ is a matrix that encodes the expected cardinality when the agents request transactions from each other in the shards corresponding to the row and column of each entry. 

\textbf{Expected shard balance:} Similarly, we can express the expected shard balance as follows:
\begin{align} \nonumber
    &E\left[1-\sum_{\lane \in \transLaneSet} \congElem_{\lane}(\transPool)\right] = \laneReqWeightVec_{ij}^\top \potMatrixBal \laneSendWeightVec_{ij} = \\ 
    &\begin{bmatrix}
        \laneReqWeight_{ij1} & \cdots & \laneReqWeight_{ij\numLanes}
    \end{bmatrix}
    \begin{bmatrix}
        1 - \congElem_1 & 1 - \congElem_1 - \congElem_2 & \cdots & 1 - \congElem_1 - \congElem_\numLanes \\ 
        1 - \congElem_1 - \congElem_2 & 1 - \congElem_2 & \cdots & 1 - \congElem_2 - \congElem_\numLanes \\
        \vdots & \vdots & \ddots & \vdots \\
        1 - \congElem_1 - \congElem_\numLanes & 1 - \congElem_2 - \congElem_\numLanes & \cdots & 1 - \congElem_\numLanes
    \end{bmatrix}
    \begin{bmatrix}
        \laneSendWeight_{ij1} \\ \vdots \\ \laneSendWeight_{ij\numLanes}
    \end{bmatrix}, 
    \label{eq:expCong}
\end{align}
where $\potMatrixBal$ encodes the shard balance values corresponding to the shards used in the transaction (we omit the argument $\transPool$ for a more compact expression).

\textbf{Expected transaction efficiency:} 
Based on the definition of transaction efficiency for the entire transaction pool \eqref{eq:transEff}, we define the efficiency of a single transaction in a given pool as the shard balance of the transaction shards divided by the cardinality: 
\begin{equation} \label{eq:singleTransEff}
    \transEff_\transPool(\trans) := \frac{\laneEff_\transPool(\transLaneSet)}{|\transLaneSet|}.
\end{equation}
Using \eqref{eq:singleTransEff}, we can write the expected efficiency of transactions from agent $j$ to $i$:
\begin{align} \nonumber
    &E\left[\transEff_\transPool(\trans)\right] = \laneReqWeightVec_{ij}^\top \potMatrixEff \laneSendWeightVec_{ij} = \\  
    &\begin{bmatrix}
        \laneReqWeight_{ij1} & \laneReqWeight_{ij2} & \cdots & \laneReqWeight_{ij\numLanes}
    \end{bmatrix}
    \begin{bmatrix}
        1-\congElem_1 & \frac{1 - \congElem_1 - \congElem_2}{2} & \cdots & \frac{1 - \congElem_1 - \congElem_\numLanes}{2} \\ 
        \frac{1 - \congElem_1 - \congElem_2}{2} & 1 - \congElem_2 & \cdots & \frac{1 - \congElem_2 - \congElem_\numLanes}{2} \\
        \vdots & \vdots & \ddots & \vdots \\
        \frac{1 - \congElem_1 - \congElem_\numLanes}{2} & \frac{1 - \congElem_2 - \congElem_\numLanes}{2} & \cdots & 1 - \congElem_\numLanes
    \end{bmatrix}
    \begin{bmatrix}
        \laneSendWeight_{ij1} \\ \laneSendWeight_{ij2} \\ \vdots \\ \laneSendWeight_{ij\numLanes}
    \end{bmatrix}.
    \label{eq:expTransEff}
\end{align}
\section{Transaction pricing}  \label{game}

We propose the use of transaction fees to align the individual goals of minimizing fees with the system-wide goal of maximizing throughput.
A natural choice is to make the fee proportional to the desired objective, which we have quantified as the transaction efficiency.
Hence, we propose the following pricing function:
\begin{equation} \label{eq:price}
     \priceFunc(\trans,\transPool) := \opPrice(\trans) + \left(1 - \frac{\laneEff_\transPool(\transLaneSet)}{|\transLaneSet|^\cardExp}\right)\maxFee,
\end{equation}
where $\opPrice(\trans)$ is the nominal transaction price, which can vary with computational requirements and market demand, $\maxFee$ is the maximum transaction fee, and $\cardExp$ is a free parameter that can be used to calibrate the  cardinality estimate or to further discourage multi-shard transactions. 
Note that the price includes one minus the transaction efficiency since the price should be low when the efficiency is high.
To simplify the remaining analysis, we assume that $\opPrice(\trans) = 0$ and $\maxFee = 1$ unless otherwise stated.
However, it is straightforward to extend the analysis to include these parameters.

While \eqref{eq:price} defines the price for a particular transaction, the expected price of a future transaction requested by agent $i$ from agent $j$ is given by:
\begin{align} \nonumber
    &E\left[\transPrice_{ij}\right] =
    1 - \laneReqWeightVec_{ij}^\top \potMatrix \laneSendWeightVec_{ij} = 1 - \\
    &\begin{bmatrix}
        \laneReqWeight_{ij1} & \laneReqWeight_{ij2} & \cdots & \laneReqWeight_{ij\numLanes}
    \end{bmatrix}
    \begin{bmatrix}
        1-\congElem_1 & \frac{1 - \congElem_1 - \congElem_2}{2^\cardExp} & \cdots & \frac{1 - \congElem_1 - \congElem_\numLanes}{2^\cardExp} \\ 
        \frac{1 - \congElem_1 - \congElem_2}{2^\cardExp} & 1 - \congElem_2 & \cdots & \frac{1 - \congElem_2 - \congElem_\numLanes}{2^\cardExp} \\
        \vdots & \vdots & \ddots & \vdots \\
        \frac{1 - \congElem_1 - \congElem_\numLanes}{2^\cardExp} & \frac{1 - \congElem_2 - \congElem_\numLanes}{2^\cardExp} & \cdots & 1 - \congElem_\numLanes
    \end{bmatrix}
    \begin{bmatrix}
        \laneSendWeight_{ij1} \\ \laneSendWeight_{ij2} \\ \vdots \\ \laneSendWeight_{ij\numLanes}
    \end{bmatrix}. 
    \label{eq:expPrice}
\end{align} 

The expression \eqref{eq:expPrice} provides a direct link from the transaction price \eqref{eq:price} to the agents' optimal choice of shards in which to request transactions for the case of agents that wish to minimize long-term expected transaction fees ($\greed_r=\greed_s=0$), allowing us to rewrite \eqref{eq:laneOptRec} as follows:
\begin{equation} \label{eq:brUpdate}
    \brLaneReqWeightVec_{ij} := \arg\max_{\laneReqWeightVec_{ij} \in \simplex_m} \laneReqWeightVec_{ij}^\top \potMatrix \laneSendWeightVec_{ij},
\end{equation} 
where $\brLaneReqWeightVec_{ij}$ denotes an update to the shard request distribution $\laneReqWeightVec_{ij}$ and $\simplex_m$ denotes the $\numLanes$-dimensional probability simplex.
The update \eqref{eq:brUpdate} is indeed a \textit{best response} of agent $i$ (in mixed-strategy space) to the mixed strategy of agent $j$.
Similarly, the sender's optimal shard choice distribution is given by
\begin{equation} \label{eq:brSendUpdate}
    \brLaneSendWeightVec_{ij} := \arg\max_{\laneSendWeightVec_{ij} \in \simplex_m} \laneReqWeightVec_{ij}^\top \potMatrix \laneSendWeightVec_{ij},
\end{equation} 
Since the sender also seeks to minimize long-term expected transaction fees, this distribution update is indeed independent from the choice of the receiver.

There is an important class of multi-player games called \textit{potential games}, in which players choose actions to maximize their individual utility functions, which in turn increases some global utility function \cite{monderer1996potential}.
A key property of potential games is that when agents act to improve their utility functions, the system is guaranteed to converge to a Nash equilibrium, which is a state in which no action by any single agent will increase their utility.
In potential games, Nash equilibria also correspond to maxima of the global utility function.
We define potential games in precise terms below.

Let $\actionSet_i$ denote the space of actions for a single agent and let $\allActions := \actionSet_i^\numAgents$ denote the set of all actions in the system.
Denote by $\utilFun_i :$ the utility function of agent $i$.
Together these define a game for the $\numAgents$ agents, and such a game is said to be a \textit{potential game} if there exists a function $\potFunc : \allActions \rightarrow \reals$ such that for any agent $i \in \agentSet$ and any pair of actions $\action, \action' \in \actionSet_i$ (where $\action_{-i}$ denotes the actions of all agents except $i$):
\begin{equation*}
    \potFunc(\action_i',\action_{-i}) - \potFunc(\action_i,\action_{-i}) = \utilFun_i(\action_i',\action_{-i}) - \utilFun_i(\action_i,\action_{-i}).
\end{equation*}

Indeed we can show that the proposed pricing mechanism induces a potential game over the transaction edges with the following potential function:
\begin{equation} \label{eq:potential}
    \potFunc  := \sum_{(k,l)\in\edgeSet} \laneReqWeightVec_{kl}^\top \potMatrix \laneSendWeightVec_{kl}.
\end{equation}

\begin{theorem} \label{thm:potential}
    The game where $n$ agents connected by the edges $\edgeSet$ update their shard request and sending distributions according to the edge-utility functions $\utilFun_{ij} := \laneReqWeightVec_{ij}^\top \potMatrix \laneSendWeightVec_{ij}$ is a potential game with the potential function \eqref{eq:potential}.
\end{theorem}

\begin{proof}
    Given an edge $(i,j) \in \edgeSet$, suppose agent $i$ updates its shard request distribution for this edge from $\laneReqWeightVec_{ij}$ to $\laneReqWeightVec_{ij}'$.
    Then, the change in the edge utility function is equal to $\utilFun_{ij}' - \utilFun_{ij} = (\laneReqWeightVec'_{ij} - \laneReqWeightVec_{ij})^\top \potMatrix \laneSendWeightVec_{ij}$.
    The resulting change in the potential function is
    \begin{equation*} 
        \potFunc' - \potFunc  := (\laneReqWeightVec'_{ij} - \laneReqWeightVec_{ij})^\top \potMatrix \laneSendWeightVec_{ij},
    \end{equation*}
    since the only change was to agent $i$'s shard request distribution from agent $j$,
    which is exactly equal to the change in the edge utility function $\utilFun_{ij}$.
    Similarly, if agent $j$ updates its shard sending distribution for this edge from $\laneSendWeightVec_{ij}$ to $\laneSendWeightVec_{ij}'$, then the change in the edge utility function is equal to $\potFunc_{ij}' - \potFunc_{ij} = \laneReqWeightVec^\top (\laneSendWeightVec'_{ij} - \laneSendWeightVec_{ij})$, which is also equal to the change in the potential function, completing the proof.
\end{proof}

Although the case we have analyzed is somewhat simplified, the fact that it constitutes a potential game is important because it ensures not only that the agents' incentives are aligned with the global objective, but that rational choices by the agents will result in convergence of the system to a maximum of the global potential function, which in our case corresponds to transaction throughput.
And their choices need not be optimal -- the only requirement is that agents take actions that increase their local utility function (lower their expected future transaction price).
This means that even deterministic shard choices (pure strategy best or better responses) will lead to convergence.
For example, a pure strategy best response update is given by replacing the probability simplex $\simplex_m$ in \eqref{eq:brUpdate} with the set of all pure strategies $\pureStrategySet_m := \{\laneReqWeightVec \in \simplex_m : \laneReqWeight_{\lane} \in \{0,1\} \text{ for each } \lane\in \laneSet\}$.

A potential problem with the approach described so far is that the update \eqref{eq:brUpdate} assumes that agents know the shard request or sending distribution of their transacting neighbors along each edge, which would require additional communication between agents.
To resolve this, let's assume that agents keeps estimates ($\laneSendWeightEstVec_{ij}$ for the receiver, $\laneReqWeightEstVec_{ij}$ for the sender) of the shard request and sending distributions of each neighbor, constructed simply from the normalized histogram of past transaction requests.
The new update rules that rely only on information available to the respective agents are then:
\begin{align} \label{eq:brReqUpdateEst}
    \brEstLaneReqWeightVec_{ij} &:= \arg\max_{\laneReqWeightVec_{ij} \in \simplex_m} \laneReqWeightVec_{ij}^\top \potMatrix \laneSendWeightEstVec_{ij}, \\
    \label{eq:brSendUpdateEst}
    \brEstLaneSendWeightVec_{ij} &:= \arg\max_{\laneSendWeightVec_{ij} \in \simplex_m} \laneReqWeightEstVec_{ij}^\top \potMatrix \laneSendWeightVec_{ij}.
\end{align} 
This turns out to be an example of \textit{fictitious play} in game theory, where agents estimate strategies of other players based on empirical distributions.
Indeed, multiplayer potential games in which players act to improve their utility using fictitious play are known to converge to a Nash equilibrium \cite{marden2009joint}.

The optimizations \eqref{eq:brReqUpdateEst}-\eqref{eq:brSendUpdateEst} are readily solved via linear programming, and the deterministic (pure-strategy) case is an integer program that reduces to finding the maximum entry in a $\numLanes$-dimensional vector:
\begin{align} \nonumber
    \brEstLaneReqWeightVec_{ij} &:= \arg\max_{\laneReqWeightVec_{ij} \in \pureStrategySet_m} \laneReqWeightVec_{ij}^\top \potMatrix \laneReqWeightEstVec_{ji} \\ \label{eq:brUpdateDet}
    &= \elVec_{\lane^*}, \text{ where } \lane^* = \arg\max_{\lane \in \laneSet} \potMatrix_{\lane} \laneReqWeightEstVec_{ji},
\end{align}
where $\elVec_\lane$ refers to column $\lane$ of the $\numLanes \times \numLanes$ identity matrix, and $\potMatrix_{\lane}$ denotes row $\lane$ of the matrix $\potMatrix$.
A similar modification can be made for the sending shard update.
The computational complexity of the pure-strategy and mixed-strategy optimizations are linear and polynomial (due to the complexity of linear programming, e.g. \cite{cohen2019solving}), respectively.

\section{Simulations} \label{simulations}

In this section, we investigate the performance of a simulated blockchain ledger with transaction price \eqref{eq:price} in which agents update their shard request distributions with pure strategy best response updates \eqref{eq:brUpdateDet}.

\subsection{Ideal case}

We begin with a simple scenario to test the pricing mechanism under ideal conditions.
Suppose that 20 agents transact with two neighbors each via a ring network on a ledger with 4 shards.
Each block contains 2500 slices meaning that the maximum capacity is 10000 transactions per block.
In this scenario, agents start with an arbitrarily large initial balance ($1e6$) in one shard (such that the resources in these shards will not be depleted), staggered among the agents, and transactions of a small fixed amount ($10$) are generated randomly by rounds.
That is, each of the 38 edges in the network executes a transaction in random order, and then the process repeats in a new random sequence until 5 blocks are eventually completed.

In the figures that follow, the x-axis corresponds to the transaction index, where anytime one shard reaches maximum capacity, a block is assembled (indicated by the vertical grid lines) and the transaction pool resets to zero.
The top panel shows the proportion of transactions contained in each shard, resulting in the balance value shown on the second panel. 
The third panel shows the mean transaction cardinality and the fourth shows the transaction efficiency calculated for each block.
As a baseline, Figure \ref{fig:random-20-4} shows the result of assigning transactions to random shards, modeling a standard hash-based sharding protocol.

We observe that the randomized policy achieves a reasonably even distribution among the shards, but since the mean cardinality is quite high, the efficiency is only about 50\%, yielding 25899 transactions in 5 blocks.

\begin{figure}[ht]
    \centering
    \includegraphics[width=\linewidth]{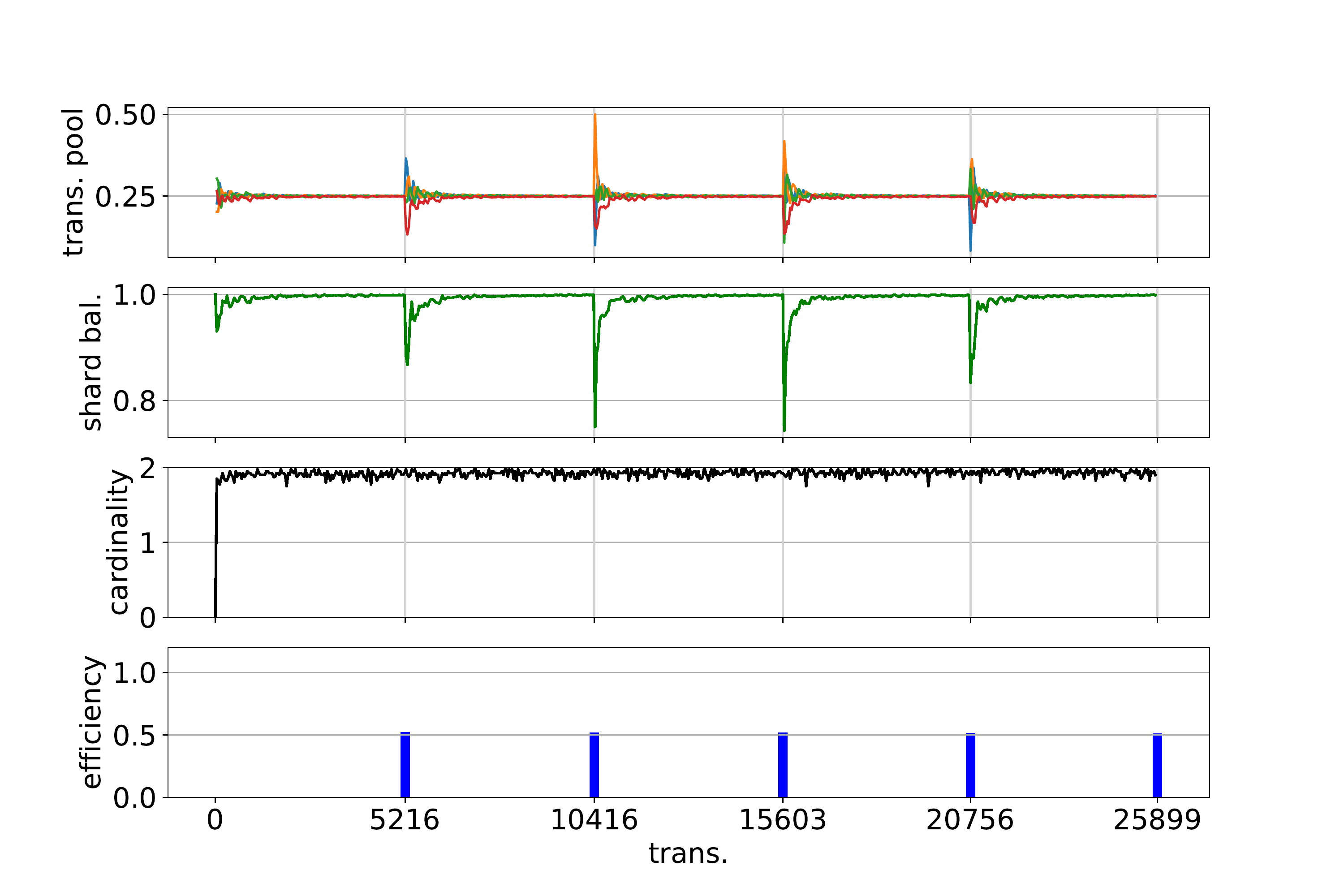}
    \caption{Fixed price and random shard requests on 20-agent ring network with 4 shards.}
    \label{fig:random-20-4}
\end{figure} 

Next, we apply load-minimizing pricing by setting $\cardExp = 0$ in \eqref{eq:price} and allowing the agents to update their shard requests with best-responses.
Figure \ref{fig:br-20-4-0} shows that the loading stays very low, but the transaction efficiency remains below 60\%, because the agents still have no incentive to transact in a single shard and therefore engage in frequent multi-shard transactions.

\begin{figure}[ht]
    \centering
    \includegraphics[width=\linewidth]{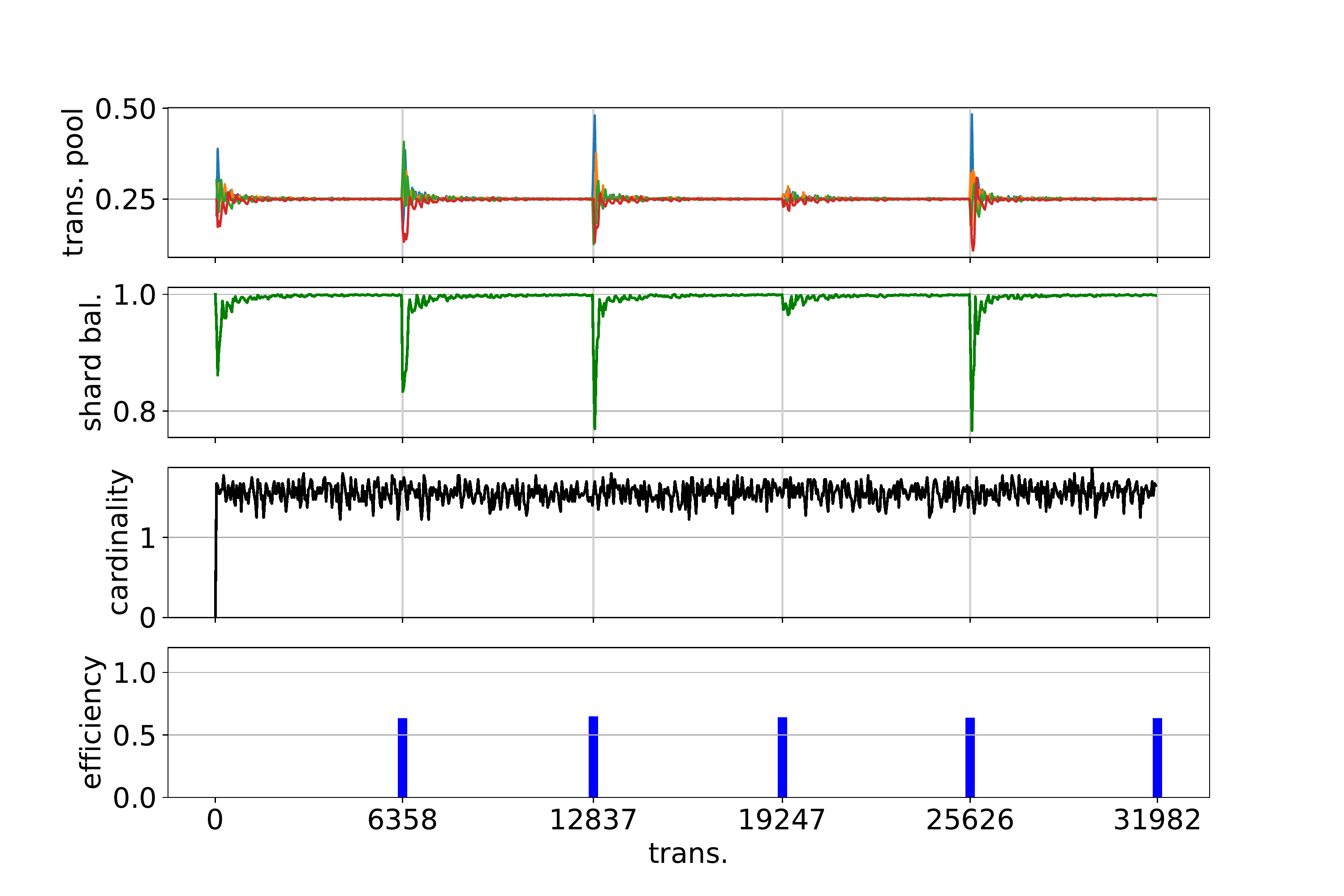}
    \caption{Congestion pricing with best-response update on 20-agent ring network with 4 shards.}
    \label{fig:br-20-4-0}
\end{figure} 

Finally, we see in Figure \ref{fig:br-20-4-1} that combining the proposed pricing mechanism ($\cardExp = 0.001$) with best-response updates results in high shard balance and low cardinality and therefore almost perfect transaction efficiency. 

\begin{figure}[ht]
    \centering
    \includegraphics[width=\linewidth]{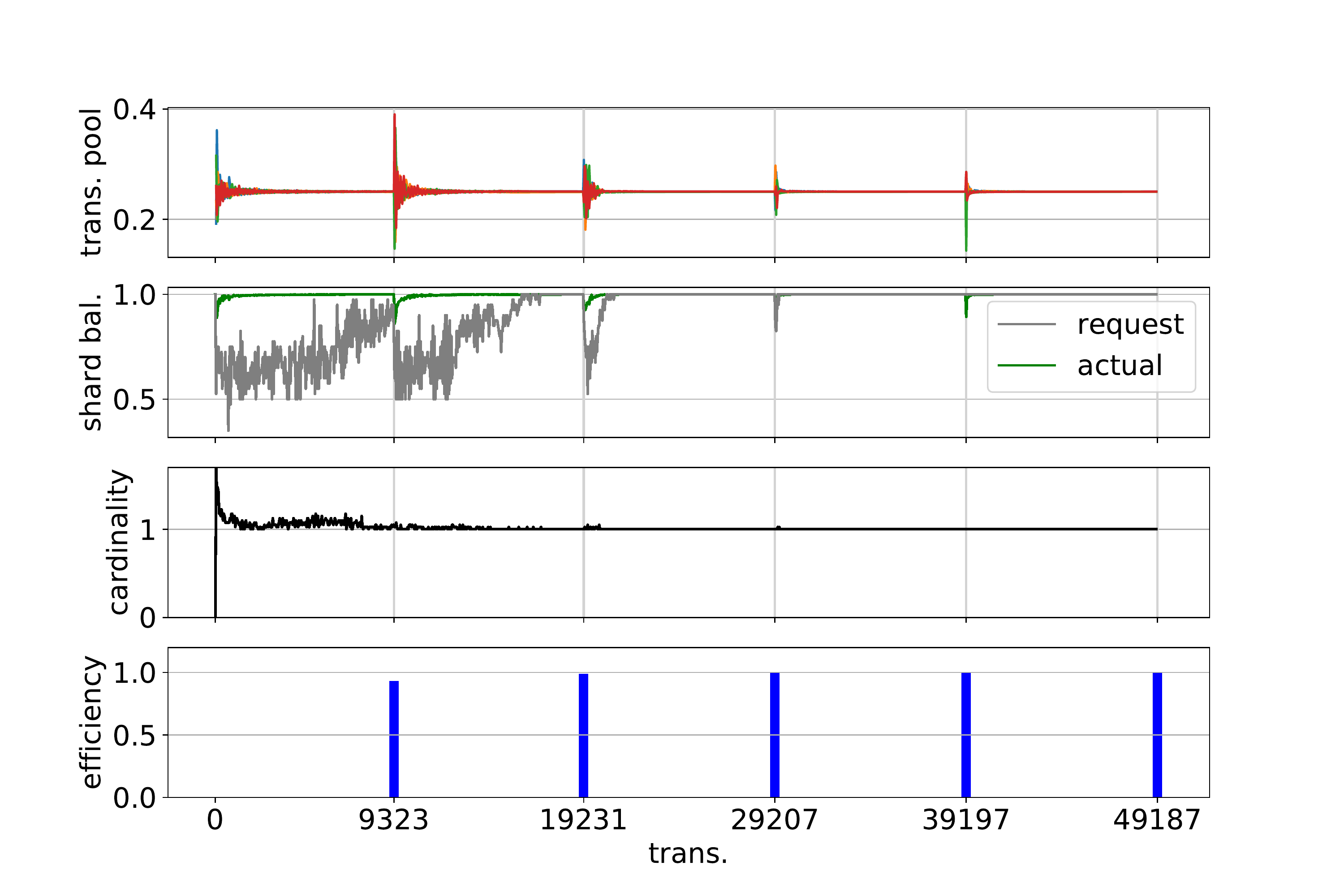}
    \caption{Efficiency pricing with best-response update on 20-agent ring network with 4 shards ($\cardExp = 0.001$).}
    \label{fig:br-20-4-1}
\end{figure}

\subsection{Larger random networks}

In the next simulation study, we construct a directed random network with a long-tailed degree distribution using a preferential attachment model.
Also, instead of generating transactions on repeated random sequences of edges, we generate them uniformly at random on the edges.
Blocks in this scenario contain 12500 slices and thus a maximum capacity of 100000 transactions.
In this case, we initialized the agents with arbitrarily large balances in each shard.
Figure \ref{fig:perf-br-100-8} shows that network still converges to a high transaction efficiency compared to the baseline shown in Figure \ref{fig:perf-random-100-8}.
The gray data in the loading plot shows how evenly distributed the shard requests are throughout the network.
As the agents acquire more information about the shard preferences of their neighbors, the shard requests converge to align with the corresponding send requests while maintaining a balance across shards.
The values of $\cardExp$ in both simulations were chosen by manual tuning, and the results are somewhat sensitive to small changes in this parameter.
An automated tuning method for $\cardExp$ would be a useful direction for future research.

\begin{figure}[ht]
    \centering
    \includegraphics[width=\linewidth]{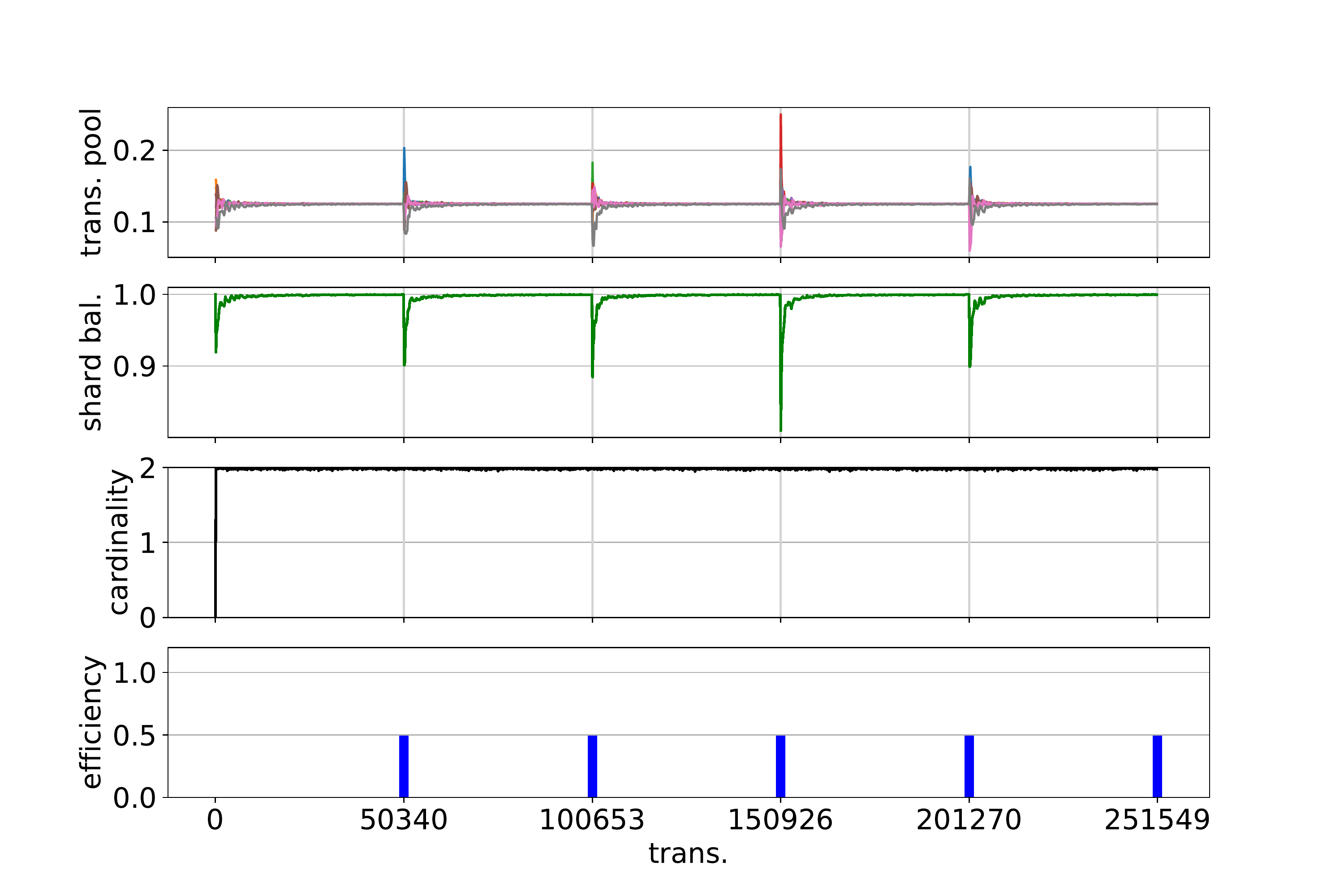}
    \caption{Fixed pricing with random shard requests on 100-agent scale-free network with 8 shards.}
    \label{fig:perf-random-100-8}
\end{figure}

\begin{figure}[ht]
    \centering
    \includegraphics[width=\linewidth]{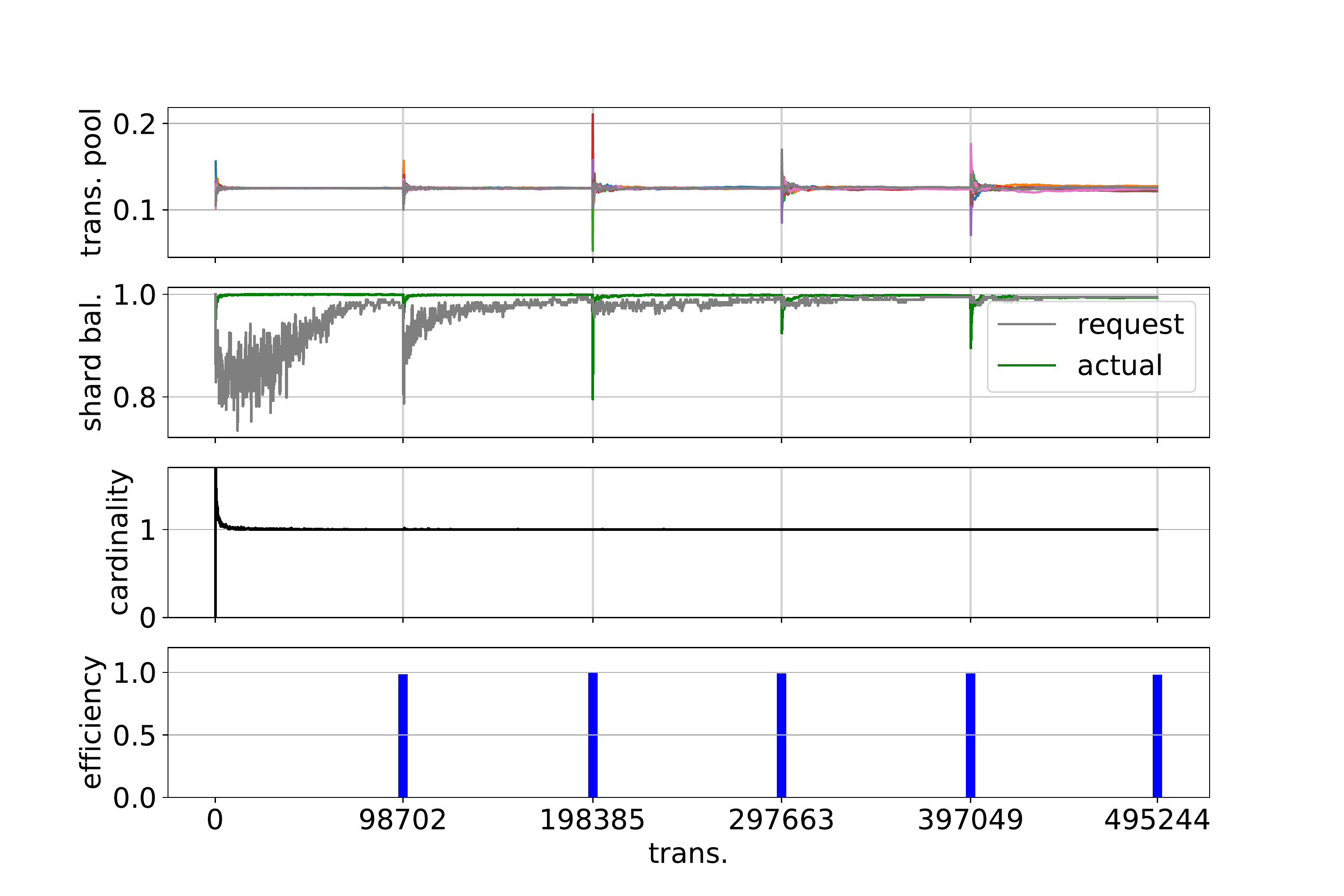}
    \caption{Efficiency pricing with best-response update on 100-agent directed random network with 8 shards ($\cardExp = 0.00015$).}
    \label{fig:perf-br-100-8}
\end{figure}

\section{Conclusions and future work} \label{conclusions}

The pricing function \eqref{eq:price} induces a potential game for the users of a sharded blockchain ledger, such that is in their interests to choose shards in a way that promotes ideal throughput conditions on the ledger.
Both the estimation of shard request distributions and the deterministic best-response policy are implementable with simple and fast computations, and guarantee convergence to a Nash equilibrium under some practical simplifying assumptions.
Simulations demonstrate the effectiveness of the policy in various conditions. 

We plan to extend the approach in several ways.
For example, there is currently no mechanism that would encourage agents to request from different neighbors in the same lane, although this would seem reasonable in practice, especially if resources are scarce. 
This could be achieved by adding a term that encourages agents to align shard requests with their resource balance.
We also intend to more explicitly account for varying costs of transaction and smart contract execution and storage.

\begin{ack}
We are grateful to Marcin Abram and Jin-Mann Wong for insightful technical discussions and for help in editing the paper.
\end{ack}

\bibliography{references}

\end{document}